\documentclass[a4paper]{revtex4-1}
\usepackage{amssymb}
\usepackage{amsthm}
\usepackage{amscd}
\usepackage[mathscr]{euscript}
\usepackage[all]{xy}
\usepackage{color}
\usepackage[utf8]{inputenc}
\normalfont
\usepackage[T1]{fontenc}
\usepackage[textwidth=14cm,hcentering]{geometry}
\usepackage[colorlinks=true,linkcolor=red,citecolor=blue]{hyperref}
\usepackage{enumerate}
\renewcommand{\l}{\ell}
\usepackage{braket}
\usepackage{stackrel}
\usepackage{hyperref}
\usepackage{graphicx}
\usepackage{mdwlist}
\usepackage{array}
\usepackage{amsmath}

\input xy
\xyoption{all}
\usepackage[table]{xcolor}

\hypersetup{colorlinks, citecolor=black, filecolor=black,linkcolor=black,urlcolor=black} 

\newtheorem{prop}{Proposition}[section]
\newtheorem{theo}[prop]{Theorem}
\newtheorem{lemm}[prop]{Lemma}

\theoremstyle{definition}
\newtheorem{defi}[prop]{Definition}

\newtheorem{exam}[prop]{Example}

\newcommand{\CC}{\mathbb{C}}

\newcommand{\PP}{\mathbb{P}}

\newcommand{\Aa}{\mathcal{A}}
\newcommand{\Bb}{\mathcal{B}}

\newcommand{\Jj}{\mathcal{J}}

\newcommand{\Oo}{\mathcal{O}}

\renewcommand{\l}{\ell}
\newcommand{\lra}{\longrightarrow}
\newcommand{\ov}[1]{\overline{#1}}

\begin{document}
 \title {Characterization of quantum entanglement\\
 via a hypercube of Segre embeddings
 }

 \author{Joana Cirici}
 \email{jcirici@ub.edu}
\address{Departament de Matemàtiques i Informàtica\\ Universitat de Barcelona\\
Gran Via 585, 
08007 Barcelona}

\author{Jordi Salvad\'{o}}
\email{jsalvado@icc.ub.edu}
\author{Josep Taron}
\email{taron@fqa.ub.edu}
\address{Departament de Fis\'ica Qu\`antica i Astrof\'isica and Institut de Ci\`encies del Cosmos\\ Universitat de Barcelona
\\Mart\'i i Franqu\`es 1,
  08028 Barcelona}

\begin{abstract}
A particularly simple description of separability of quantum states
arises naturally in the setting of complex algebraic geometry, via the
Segre embedding. This is a map describing how to take products of
projective Hilbert spaces. In this paper, we show that for pure states
of $n$ particles, the corresponding Segre embedding may be described
by means of a directed hypercube of dimension $(n-1)$, where all edges
are bipartite-type Segre maps. Moreover, we describe the image of the
original Segre map via the intersections of images of the $(n-1)$
edges whose target is the last vertex of the hypercube. This purely
algebraic result is then transferred to physics. For each of the last
edges of the Segre hypercube, we introduce an observable
which measures geometric
separability and is related to the trace of the squared reduced
density matrix. As a consequence, the hypercube approach gives a novel
viewpoint on measuring entanglement, naturally relating bipartitions
with $q$-partitions for any $q\geq 1$.  
We test our observables against well-known states, showing that these provide well-behaved and fine measures of entanglement.
\end{abstract}

\maketitle
\setcounter{tocdepth}{1}
%

\section{Introduction}
Quantum entanglement is at the heart of quantum physics,
with crucial roles in quantum information theory, superdense coding
and quantum teleportation among others.

An important problem in entanglement theory is to obtain separability criteria. While there is a clear definition of separability, in general it is difficult to determine whether a given state is entangled or separable. 
A refinement of this problem is to quantify entanglement on a given entangled state. This is a broadly open problem, in the sense that there is not a unique established way of measuring entanglement, which might depend on the initial set-up and applications on has in mind.
There is, however, a general consensus on the desirable properties of a good entanglement measure \cite{1997PhRvL..78.2275V,2009PhR...474....1G}.
Directly attached to measuring entanglement is the notion of \textit{maximally entangled state}, central in teleportation protocols.

Many different entanglement measures and the corresponding notions of maximal entanglement have been proposed. Methods range from using Bell inequalities, looking at inequalities in the larger scheme of entanglement witnesses, spin squeezing inequalities, entropic inequalities, the measurement of nonlinear properties of the quantum state or the approximation of positive maps.
We refer to the exhaustive reviews
\cite{Horodecki:2009zz, 2009PhR...474....1G, Plenio} for the basic aspects of entanglement including its history, characterization, measurement, classification and applications.

A particularly simple description of separability of quantum states arises naturally in the setting of complex algebraic geometry.
In this setting, pure multiparticle states are identified with points in the complex projective space $\PP^N$, the set of lines of the complex space $\CC^{N+1}$ that go through the origin.
Entanglement is then understood via the categorical product of projective spaces: the \textit{Segre embedding}. This is a map of complex algebraic varieties
\begin{equation}\label{GenSeg}
\PP^1\times\stackrel{(n)}{\cdots}\times\PP^1\lra \PP^{2^n-1},\end{equation}
whose image, called the \textit{Segre variety}, is described in terms of a family of homogeneous quadratic polynomial equations in $2^n$ variables, where $n$ is the number of particles. The points of the Segre variety correspond precisely to separable states (see for instance \cite{Ballico,Bengtsson}).

In this paper, we exploit this geometric viewpoint to show that the image 
of (\ref{GenSeg})
is in fact given by the intersection of all Segre varieties defined via bipartite-type Segre maps
\begin{equation}\label{BipSeg}
\PP^{2^\ell-1}\times\PP^{2^{n-\ell}-1}\lra \PP^{2^n-1},
\text{ for all }1\leq \ell\leq n-1.
\end{equation}
Specifically, we show that a state is $q$-partite if and only if it
lies in $q$ of the images of Segre maps of the form (\ref{BipSeg}). We
do this after showing that the Segre embedding (\ref{GenSeg}) may be
decomposed in various equivalent ways, leading to a hypercube of
dimension $(n-1)$ whose edges are bipartite-type Segre maps. In this framework all the information about the separability of the state in contained
in the last applications (\ref{BipSeg}) of the hypercube. 
There are numerous approaches to entanglement via Segre varieties that are related to the present work \cite{BHL,BBZ,Heydari3,Heydari2,Heydari,Grabowski,Sanz}. The hypercube viewpoint presented here is a novel approach that connects the notions of bipartite and $q$-partite in a geometric way.

While the above results are extremely precise and intuitive, they are purely algebraic. 
In order to build a bridge from geometry to physics, we introduce a
family of observables $\{\Jj_{n,\ell}\}$, with $1\leq\ell\leq n-1$,
which allow us to detect when a given $n$-particle state belongs to
each of the images of the bipartite-type Segre maps (\ref{BipSeg}). As
a consequence, we obtain that a state  is $q$-partite if and only if
at least $q$ of the observables $\Jj_{n,\ell}$ vanish on this state, completely identifying the sub-partitions of the system.
Each of these observables are related to the trace of the squared reduced
density matrix, also known as the bipartite non-extensive Tsallis entropy with entropic index two \cite{2006JMP....47b3502H,2004CMT....16..223T}.
They are always positive on entangled states and our first applications indicate that they provide well-behaved and fine measures of entanglement.

\medskip 

We briefly explain the contents of this paper. We begin with a warm-up Section \ref{Sec2part} where we detail the theory and results for the well-understood settings of two- and three-particle states.
In Section \ref{SecGeom} we develop the geometric aspects of the paper. In particular, we describe the Segre hypercube and prove Theorem \ref{entanglement_Segre} on geometric decomposability. The main result of Section \ref{SecOps} is Theorem \ref{teomeasure}, where we match geometric decomposability with our family of observables related to non-extensive entropic measures. The two theorems are combined in Section \ref{Secphysical}, where we study entanglement measures for pure multiparticle states of spin-${1\over 2}$ and 
apply our observables on various well-known multiparticle states.

\begin{acknowledgments} We would like to thank Joan Carles Naranjo for his ideas in the proof of Lemma \ref{pullback}.\\

J. Cirici would like to acknowledge partial support from the AEI/FEDER, UE (MTM2016-76453-C2-2-P) and the Serra H\'{u}nter Program.
J. Salvad\'{o} and J. Taron are partially supported by the Spanish grants
FPA2016-76005-C2-1-PEU, PID2019-105614GB-C21, PID2019-108122GB-C32, by the Maria de Maeztu grant MDM-2014-0367 of
ICCUB, and by the European INT projects FP10ITN ELUSIVES (H2020-MSCA-ITN-2015-674896) and
INVISIBLES-PLUS (H2020-MSCA-RISE-2015-690575).
\end{acknowledgments}

\section{Warm-up: two and three particle states}\label{Sec2part}
We begin with the well-understood example of two particle entanglement.
The initial set-up consists in two particles which can be shared between two
different observers, $A$ and $B$, that can perform quantum measures to each of the particles. 
A general pure state for two spin-$\frac{1}{2}$ particles can be written as
\begin{equation}\label{psiAB}
\ket{\psi}_{AB} = z_{0}\ket{00} + z_{1}\ket{01} +z_{2}\ket{10}+ z_{3}\ket{11},
\end{equation}
where $z_{i}$ are complex numbers that satisfy the normalization condition
$\sum|z_{i}|^2=1$. The labels $AB$, which will be often omitted, indicate that $A$ is acting on
the first particle while $B$ is acting on the second, so 
\[\ket{ij}:=\ket{i}_A\otimes\ket{j}_B,\,\text{ for }i,j\in\{0,1\}.\]

Entanglement is a property that can be inferred from statistical
properties of different measurements by the observers of the
system. In particular, we can compute the sum of the expected values of
$A$ measuring the spin of the state $\ket{\psi}$ in the
different directions. With this idea in mind, we define:
\[\Jj_{A\otimes B}(\psi):=2-\left(\sum_{i=0}^3 |\bra{\psi}\sigma_i\otimes \mathbb{I}_2\ket{\psi}|^2\right),\]
where $\sigma_i$, for $i=1,2,3$, denote the Pauli matrices, $\sigma_0=\mathbb{I}_2$ is the identity matrix of size two and $\otimes$ denotes the Kronecker product. 
Using the expression (\ref{psiAB}) for $\ket{\psi}$ we obtain
\[\Jj_{A\otimes B}(\psi)=4|z_{0} z_{3} - z_{1} z_{2}|^2.\]
It is well-known that the state $\ket{\psi}$ is entangled if and only if 
\[z_0z_3-z_1z_2=0.\]
Therefore we find that $\ket{\psi}$ is a product state if and only if $\Jj_{A\otimes B}(\psi)=0$.
When $\Jj_{A\otimes B}(\psi)>0$ we have an entangled state, which is considered to be maximally entangled when the observable reaches its maximum value at 
$\Jj_{A\otimes B}(\psi)=1$. 

The measure given by $\Jj_{A\otimes B}(\psi)$ may also be interpreted
in terms of the density matrix operator. Indeed, the density matrix $\rho$ for a
pure state $\ket{\psi}$ is given by
\[\rho = \ket{\psi}\bra{\psi} =
\left( \begin{array}{cccc}
  z_0\ov z_0 &  z_0\ov z_1 & z_0\ov z_2 & z_0\ov z_3\\
  z_1\ov z_0 &  z_1\ov z_1 & z_1\ov z_2 & z_1\ov z_3\\
  z_2\ov z_0 &  z_2\ov z_1 & z_2\ov z_2 & z_2\ov z_3\\
  z_3\ov z_0 &  z_3\ov z_1 & z_3\ov z_2 & z_3\ov z_3 \end{array} \right).\]
  Compute the reduced matrix
for the subsystem $A$ 
by means of the
partial trace on $B$, 
\[\rho_A = {\rm Tr}_B \rho =
\left( \begin{array}{cc}
  z_0\ov z_0 + z_1\ov z_1 & z_0\ov z_2 + z_1\ov z_3\\
  z_2\ov z_0 + z_3\ov z_1 & z_2\ov z_2 + z_3\ov z_3 \end{array} \right).\]
The trace in the subsystem $A$ of the square reduced matrix
gives
\[{\rm Tr} \rho_A^2 = 1-2 |z_{0} z_{3} - z_{1} z_{2}|^2 .\]
This magnitude is refereed in the literature as the Tsallis entropy or $q$-entropy
with $q=2$ and is directly related with the proposed measure by,
\[{\cal J}_{A \otimes B }(\psi)=2 \left( 1 - {\rm Tr_A} \rho_A^2 \right).\]

Recall that by the Schmidt Decomposition, and after choosing a convenient basis, any pure 2-particle state may be written as 
\[\ket{\psi}=\cos(\theta)\ket{00}+\sin(\theta)\ket{11}\]
where $\theta\in[0,\pi/4]$ is the \textit{Schmidt angle}, which is known to quantify entanglement (see for instance \cite{Peres}).
A simple computation gives
\[\Jj_{A\otimes B}(\psi)=4(\cos^2(\theta)\sin^2(\theta)).\]
The two well-known states 
\[\ket{\mathrm{Sep}}=\ket{00}\text{ and }\ket{\mathrm{EPS}}={1\over \sqrt{2}}(\ket{00}+\ket{11}),\]
which can be taken as representatives for the two only classes of states under the action of Stochastic Local Operations and Classical Communication (SLOCC),
correspond to the two extreme cases of separated ($\theta=0$) and maximally entangled ($\theta=\pi/4$) states respectively,
in the sense that 
$\ket{\mathrm{EPS}}$ gives the greatest violation of Bell inequalities, has the largest entropy of entanglement,
and its one-party reduced states are both maximally mixed.
\medskip

The characterization of entanglement has a simple geometric interpretation.
Note first that a general pure state for a single spin-$\frac{1}{2}$ particle may be written as 
\[z_{0}\ket{0} + z_{1}\ket{1}\text{ where }|z_0|^2+|z_1|^2=1.\] 
In particular, such a state is determined by the pair of complex numbers $(z_0,z_1)$ and the normalization condition ensures $(z_0,z_1)\neq(0,0)$. This allows one to consider the corresponding equivalence class $[z_0:z_1]$ in the complex projective line $\PP^1$. 
This space is defined as the set of lines of $\CC^2$ that go through the origin.
The class $[z_0:z_1]$ denotes the the set of all points $(z_0',z_1')\in\CC^2\setminus\{(0,0)\}$ such that there is a non-zero complex number $\lambda$ with $(z_0,z_1)=\lambda(z_0',z_1')$.
Note that, by construction, we have $[z_0:z_1]=[\lambda z_0:\lambda z_1]$ for all $\lambda\in\CC^*$.

In summary, every one particle state of spin-$\frac{1}{2}$ 
defines a unique point in $\PP^1$. Conversely, given a point $[z_0:z_1]\in \PP^1$ we may choose a representative $(z_0,z_1)$ such that $|z_0|^2+|z_1|^2=1$ and so it determines a unique pure state 
$z_0\ket{0}+z_1\ket{1}$
up to a global phase which does not affect any state measurements.

This one-to-one correspondence between pure states and points in the projective space generalizes analogously to several particles: pure states for $n$ particles of spin-$\frac{1}{2}$
correspond to points in $\PP^{2^n-1}$. This correspondence is just a way of describing the projectivization of the Hilbert space of quantum states and so is  valid for particles of arbitrary spin, after adjusting dimensions accordingly.

For our two-particle case, since the
state $\ket{\psi}$ introduced in (\ref{psiAB}) is determined by the normalized set of complex numbers $(z_{0},
z_{1},z_{2},z_{3})$ we obtain a point $[\psi]:=[z_{0}:z_{1}:z_{2}:z_{3}]$ in the complex projective space $\PP^3$, the set of lines in $\CC^4$ going through the origin. Interestingly for the study of entanglement, there is a map 
\[f_{A\otimes B}:\PP^1_A\times \PP^1_B\lra\PP^3_{AB},\]
called the \textit{Segre embedding}, defined by the products of coordinates
\[[a_0:a_1],[b_0:b_1]\mapsto [a_0b_0:a_0b_1:a_1b_0:a_1b_1].\]
The Segre map is the categorical product of projective spaces, describing how to take products on projective Hilbert spaces. The word \textit{embedding} accounts for the fact that this map is injective: it embeds the product $\PP^1\times \PP^1$, which has complex dimension 2, inside $\PP^3$, which has complex dimension 3.
The image of this map 
\[\Sigma_{A\otimes B}:=\mathrm{Im}(f_{A\otimes B})\]
is called the \textit{Segre variety}. This is a complex algebraic variety of dimension $2$ 
and is given by the set of points $[z_{0}:z_{1}:z_{2}:z_{3}]\in \PP^3$
satisfying the single quadratic polynomial equation
\[z_{0} z_{3}-z_{1} z_{2}=0.\]
In particular, product states correspond precisely to points in the Segre variety
and 
\[\Jj_{A\otimes B}(\psi)=0\Longleftrightarrow [\psi]\in \Sigma_{A\otimes B}.\]
As a consequence, $\ket{\psi}_{AB}$ is a product state if and only if its
corresponding point $[\psi]$ in $\PP^3$ lies in the Segre variety $\Sigma_{A\otimes B}$.

\medskip 

The entanglement of three particle states is also well-understood in the literature \cite{2000PhRvA..62f2314D}.
However, this case already exhibits some non-trivial facts that arise in the geometric interpretation of entanglement. We briefly review this case before discussing the general set-up.

A general pure state for three spin-$\frac{1}{2}$ particles can be written as
\begin{equation}\label{psi3}
\def\arraystretch{1.6}
\begin{array}{ll}
\ket{\psi}_{ABC}& = z_{0}\ket{000} + z_{1}\ket{001} +z_{2}\ket{010}+
z_{3}\ket{011}+\\
& +z_{4}\ket{100}+ z_{5}\ket{101}+ z_{6}\ket{110}+ z_{7}\ket{111}
\end{array}.
\end{equation}
As in the two particle case, $z_{i}$ are complex numbers that satisfy the normalization condition
$\sum|z_{i}|^2=1$. Now we have a third observer, $C$, in addition to $A$ and $B$, so
\[\ket{ijk}:=\ket{i}_A\otimes\ket{j}_B\otimes \ket{k}_C\text{ for }i,j,k\in\{0,1\}.\] 

We will first measure bipartitions: the separability of this state into states of the form 
\[\ket{\varphi}_A\otimes\ket{\varphi'}_{BC}\text{ or }\ket{\varphi}_{AB}\otimes\ket{\varphi'}_{C}.\]
For the first case, we define the observable
\begin{align*}
\Jj_{A\otimes BC}(\psi):=2-\left(\sum_{i=0}^3 |\bra{\psi}\sigma_i\otimes \mathbb{I}_4\ket{\psi}|^2\right)=\\
=4\left\{
|z_0z_5-z_1z_4|^2+|z_0z_6-z_2z_4|^2+|z_0z_7-z_3z_4|^2+\right.\\
\left.+|z_1z_6-z_2z_5|^2+|z_1z_7-z_3z_5|^2+|z_2z_7-z_3z_6|^2
\right\},\,\end{align*}
where the last equality will be proven in Section \ref{SecOps}.
Let us for now interpret this expression geometrically. Our state $\ket{\psi}$ corresponds to a point in $\PP^7$ and bipartitions of type $A\otimes BC$ are geometrically characterized by the Segre embedding
\[f_{A\otimes BC}:\PP^1_{A}\otimes\PP^3_{BC}\lra \PP^7_{ABC}\]
defined by sending the tuples $[a_0:a_1]$, $[b_0:b_1:b_2:b_3]$ to the point of $\PP^7$ given by 
\[[a_0b_0:a_0b_1:a_0b_2:a_0b_3:a_1b_0:a_1b_1:a_1b_2:a_1b_3].\]
Specifically, the state $\ket{\psi}$ can be written as $\ket{\varphi}_A\otimes\ket{\varphi'}_{BC}$ if and only if the corresponding point 
$[\psi]=[z_0:\cdots:z_7]\in \PP^7$ lies in the Segre variety 
\[\Sigma_{A\otimes BC}:=\mathrm{Im}(f_{A\otimes BC}).\]
The equations defining $\Sigma_{A\otimes BC}$, having set coordinates $[z_0:\cdots:z_7]$ of $\PP^n$, are given by the vanishing of all $2\times 2$ minors of the matrix
\[
\left(
\begin{matrix}
z_0&z_1&z_2&z_3\\
z_4&z_5&z_6&z_7
\end{matrix}
\right).
\]
Therefore we see that 
\[\Jj_{A\otimes BC}(\psi)=0 \Longleftrightarrow [\psi]\in \Sigma_{A\otimes BC}.\]

For the second case, we define:
\begin{align*}
\Jj_{AB\otimes C}(\psi):=2-\left({1\over 2}\sum_{i,j=0}^3 |\bra{\psi}\sigma_i\otimes \sigma_j\otimes I_2\ket{\psi}|^2\right)=\\
=4\left\{
|z_0z_3-z_1z_2|^2+|z_0z_5-z_1z_4|^2+|z_0z_7-z_1z_6|^2+\right.\\
\left.+|z_2z_5-z_3z_4|^2+|z_2z_7-z_3z_6|^2+|z_4z_7-z_5z_6|^2
\right\},\,\end{align*}
where again, the last equality is detailed in Section \ref{SecOps}.
Bipartitions of type $AB\otimes C$ are now geometrically characterized by the Segre embedding
\[f_{AB\otimes C}:\PP^3_{AB}\otimes\PP^1_{C}\lra \PP^7_{ABC}.\]
The state $\ket{\psi}$ can be written as $\ket{\varphi}_{AB}\otimes\ket{\varphi'}_{C}$ if and only if the corresponding point 
$[\psi]=[z_0:\cdots:z_7]\in \PP^7$ lies in the Segre variety 
\[\Sigma_{AB\otimes C}:=\mathrm{Im}(f_{AB\otimes C}).\]
In this case, the Segre variety $\Sigma_{AB\otimes C}$ is given by all points $[z_0:\cdots:z_7]\in\PP^7$ such that
all $2\times 2$ minors of the matrix
\[
\left(
\begin{matrix}
z_0&z_1\\
z_2&z_3\\
z_4&z_5\\
z_6&z_7
\end{matrix}
\right)
\]
vanish.
Therefore we may conclude that 
\[\Jj_{AB\otimes C}(\psi)=0 \Longleftrightarrow [\psi]\in \Sigma_{AB\otimes C}.\]

We may now ask about separability of the state $\ket{\psi}$ in a totally decomposed form
 \[\ket{\varphi}_A\otimes \ket{\varphi'}_B\otimes \ket{\varphi''}_C.\]
It turns out that the above defined observables are sufficient in order to address this question.
This is easily seen using the geometric characterization of entanglement, as we next explain.
The total separability of the state $\ket{\psi}_{ABC}$ is geometrically characterized by the generalized Segre embedding 
\[f_{A\otimes B\otimes C}:\PP^1_A\times \PP^1_B\times\PP^1_C\lra \PP^7_{ABC}\]
defined by sending the tuples $[a_0:a_1],[b_0:b_1],[c_0:c_1]$
to the point in $\PP^7$ given by
\[[a_0b_0c_0:a_0b_0c_1:a_0b_1c_0:a_0b_1c_1:a_1b_0c_0:a_1b_0c_1:a_1b_1c_0:a_1b_1c_1].\]
The state $\ket{\psi}$ is separable as $A\otimes B\otimes C$ if and only if its associated point $[\psi]\in\PP^7$ lies in the generalized Segre variety given by 
\[\Sigma_{A\otimes B\otimes C}:=\mathrm{Im}(f_{A\otimes B\otimes C}).\]
The above generalized Segre embedding factors in two equivalent ways:
\[f_{A\otimes B\otimes C}=f_{A\otimes BC}\circ (\mathbb{I}_A\times f_{B+C})= f_{AB\otimes C}\circ(f_{A\otimes B}\times \mathbb{I}_C),\]
so we have a commutative square
\[
\xymatrix{
\PP^1_A\times \PP^1_B\times \PP^1_C\ar[d]_{\mathbb{I}_A\times f_{B+C}}
\ar[rr]^{f_{A\otimes B}\times \mathbb{I}_C}&&\PP^3_{AB}\times \PP^1_C\ar[d]^{f_{AB\otimes C}}\\
\PP^1_A\times \PP^3_{BC}\ar[rr]^{f_{A\otimes BC}}&&\PP^7_{ABC}
}.
\]
We will show (Theorem \ref{entanglement_Segre}) that the Segre variety $\Sigma_{A\otimes B\otimes C}$ 
agrees with the intersection 
\[\Sigma_{A\otimes B\otimes C}= \Sigma_{A\otimes BC}\cap \Sigma_{AB\otimes C}.\]
In particular, we have 
\[\Jj_{A\otimes BC}(\psi)=0\text{ and } \Jj_{AB\otimes C}(\psi)=0\Longleftrightarrow [\psi]\in \Sigma_{A\otimes B\otimes C}.\]

In summary, the observables $\Jj_{AB\otimes C}$ and $\Jj_{A\otimes BC}$ determine separability of any three particle state in the $ABC$ order.
Of course, one can also measure separability for the orders $BAC$ and $CAB$ by consistently taking into account permutations of the chosen Hilbert bases, as we next illustrate.

Consider the following well-known states and their corresponding points in $\PP^7$:
\[\arraycolsep=1.4pt\def\arraystretch{1.6}
\begin{array}{lll}
\ket{\mathrm{Sep}}=\ket{000}&&[\mathrm{Sep}]=[1:0:0:0:0:0:0:0]\\
\ket{B_1}={1\over \sqrt{2}}\left(\ket{000}+\ket{011}\right)&\quad&[B_1]=[1:0:0:1:0:0:0:0]\\
\ket{B_2}={1\over \sqrt{2}}\left(\ket{000}+\ket{101}\right)&&[B_2]=[1:0:0:0:0:1:0:0]\\
\ket{B_3}={1\over \sqrt{2}}\left(\ket{000}+\ket{110}\right)&&[B_3]=[1:0:0:0:0:0:1:0]\\
\ket{W}={1\over \sqrt{3}}\left(\ket{100}+\ket{010}+\ket{001}\right)&&[W]=[0:1:1:0:1:0:0:0]\\
\ket{\mathrm{GHZ}}={1\over \sqrt{2}}\left(\ket{000}+\ket{111}\right)&\quad\quad\quad\quad&[\mathrm{GHZ}]=[1:0:0:0:0:0:0:1]\\
\end{array}
\]
These can be taken as representatives for the six existing equivalence classes of three particle states under SLOCC-equivalence.
The state $\ket{\mathrm{Sep}}$ is obviously separable and  $\ket{\mathrm{GHZ}}$ and $\ket{W}$ are the only genuinely 
entangled states. Moreover, these two entangled states represent two different equivalence classes of entanglement \cite{2000PhRvA..62f2314D}.
The former state, named after \cite{1989GHZ}, is maximally entangled with respect to most entanglement measures existing in the literature and its one-particle reduced density matrices are all maximally mixed. 
The remaining states $\ket{B_i}$ are 2-partite (depending on the order of the Hilbert basis).
We have the following table:
\[
\arraycolsep=4pt\def\arraystretch{1.4}
 \begin{array}{|c|c|c|c|c|c|c|c|}
 \hline
ABC&\mathrm{Sep}&{B_1}&{B_2}&{B_3}&{W}&\mathrm{GHZ}\\
 \hline\hline
 \Jj_{A\otimes BC}&0&1&1&0&{8\over 9}&1\\
 \Jj_{AB\otimes C}&0&0&1&1&{8\over 9}&1\\
 \hline\hline 
 \Jj&0&{1\over 2}&1&{1\over 2}&{8\over 9}&1\\
 \hline
\end{array}\]
Here the labels $ABC$ indicate we are measuring entanglement of the states with the fixed order $ABC$ and 
\[\Jj:={1\over 2}(\Jj_{A\otimes BC}+\Jj_{AB\otimes C})\] is the average measure. In particular, we see that while $B_1$ and $B_3$ are bipartite in this order, the state $B_2$ is classified as entangled.
Note however that if we measure entanglement with respect to the order $ACB$,
the roles of $B_2$ and $B_3$ are exchanged and we
obtain the following table.
\[
\arraycolsep=4pt\def\arraystretch{1.4}
 \begin{array}{|c|c|c|c|c|c|c|c|}
 \hline
ACB&\mathrm{Sep}&{B_1}&{B_2}&{B_3}&{W}&\mathrm{GHZ}\\
 \hline\hline
 \Jj_{A\otimes BC}&0&1&0&1&{8\over 9}&1\\
 \Jj_{AB\otimes C}&0&0&1&1&{8\over 9}&1\\
\hline\hline 
 \Jj&0&{1\over 2}&{1\over 2}&1&{8\over 9}&1\\
 \hline
\end{array}\]

The states $\ket{\mathrm{Sep}}$, $\ket{W}$ and $\ket{\mathrm{GHZ}}$
are invariant under permutations of the basis $ABC$ and so the values of the observables always remain unchanged.
Note as well that $\ket{\mathrm{W}}$ exhibits less entanglement than $\ket{\mathrm{GHZ}}$ with respect to the above measures, in agreement with the existing entanglement measures.

\section{Hypercube of Segre embeddings}\label{SecGeom}
This section is purely mathematical. Given an integer $n\geq 2$, we consider the generalized Segre embedding
\[\PP^1\times\stackrel{(n)}\cdots\times \PP^1\lra \PP^{2^n-1}\]
and introduce the notion of \textit{$q$-decomposability} of a point $z\in \PP^{2^n-1}$ for any integer $1<q\leq n$. We show that $q$-decomposability is detected by looking at all the Segre embeddings of the type 
\[\PP^{2^\ell-1}\times \PP^{2^{n-\ell}-1}\lra \PP^{2^n-1},\text{ for }1\leq \ell\leq n-1,\]
which accommodate as edges of a directed hypercube of Segre embeddings.
Let us first review some basic definitions and constructions.

The \textit{complex projective space} $\PP^N$ is  
the set of lines in the complex space $\CC^{N+1}$ passing through the origin. It may be described as the quotient
\[\PP^N:={\CC^{N+1}- \{0\}\over z\sim \lambda z},\, \lambda \in \CC^*.\]
A point $z\in \PP^N$ will be denoted by its homogeneous coordinates
$z=[z_0:\cdots:z_N]$ where, by definition,
there is always at least an integer $i$ such that $z_i\neq 0$, and for any $\lambda\in\CC^*$ we have 
\[[z_0:\cdots:z_N]=[\lambda z_0:\cdots:\lambda z_N].\]

\begin{defi}
Given positive integers $k$ and $\l$, the \textit{Segre embedding} $f_{k,\l}$ is the map
\[f_{k,\l}:\PP^k\times \PP^\l\lra \PP^{(k+1)(\l+1)-1}\]
defined by sending a pair of points $a=[a_0:\cdots:a_k]\in\PP^k$ and $b=[b_0:\cdots:b_\l]\in\PP^\ell$
to the point of $\PP^{(k+1)(\l+1)-1}$ whose homogeneous coordinates
are the pairwise products of the homogeneous coordinates of $a$ and $b$:
\[f_{k,\l}(a,b)=  [\cdots:z_{ij}:\cdots]\text{ with }z_{ij}:=a_ib_j,\]
where we take the lexicographical order.
\end{defi}
The Segre embedding is injective, but not surjective in general. 
The image of $f_{k,\l}$ is called the \textit{Segre variety} and is denoted by
\[\Sigma_{k,\l}:=\mathrm{Im}(f_{k,\l})=\left\{[\cdots:z_{ij}:\cdots]\in \PP^{(k+1)(\l+1)-1}; z_{ij}z_{i'j'}-z_{ij'}z_{i'j}=0,\,\forall\, i\neq i' j\neq j'\right\}.\] 
In other words, $\Sigma_{k,\l}$ is given by the zero locus of all the $2\times 2$ minors of the matrix
\[\left(
\begin{matrix}
 z_{00}&\cdots&z_{0\l}\\
 \vdots&\ddots&\vdots\\
 z_{k0}&\cdots&z_{k\l}
\end{matrix}\right).
\]
A combinatorial argument shows that there is a total of 
\[\xi_{k,\ell}:=\left(\begin{matrix}
                       k+1\\ 2
                      \end{matrix}\right)\cdot 
      \left(\begin{matrix}
                       \ell+1\\ 2
      \end{matrix}\right)= {k\cdot(k+1)\cdot \ell\cdot (\ell+1)\over 4 }
      \]
minors of size $2\times 2$ in such a matrix.

%
%

The above construction generalizes to products of more than two projective spaces of arbitrary dimensions as follows.
Given positive integers $k_1,\cdots,k_n$, 
let
 \[N(k_1,\cdots,k_n):=(k_1+1)\cdots(k_n+1)-1.\]
 For $1\leq j\leq n$, let $[a_{0}^{j}:\cdots:a_{k_j}^j]$ denote coordinates of $\PP^{k_j}$.
 \begin{defi}\label{defgenseg}
The \textit{generalized Segre embedding}
\[f_{k_1,\cdots,k_n}:\PP^{k_1}\times\cdots \times \PP^{k_n}\lra \PP^{N(k_1,\cdots,k_n)}\]
is defined by letting 
\[f_{k_1,\cdots,k_n}([\cdots:a^1_{i_1}:\cdots],\cdots,[\cdots:a^n_{i_n}:\cdots]):= [\cdots:z_{i_1\cdots i_n}:\cdots]
\text{ where }z_{i_1\cdots i_n}=a^1_{i_1}\cdots a^n_{i_n}
\]
and the lexicographical is assumed. Denote the \textit{generalized Segre variety} by 
\[\Sigma_{k_1,\cdots,k_n}:=\mathrm{Im}(f_{k_1,\cdots,k_n}).\]
 \end{defi}
 
It follows from the definition that every
generalized Segre embedding may be written as compositions of maps of the form
\[\mathbb{I}_{m}\times f_{k,\l}\times \mathbb{I}_{m'}\]
for certain values of $m$, $k$, $\l$ and $m'$, where $\mathbb{I}_m$ denotes the identity map of $\PP^m$.
These compositions may be arranged in a directed $(n-1)$-dimensional hypercube, where the initial vertex 
is $\PP^{k_1}\times\cdots\times\PP^{k_n}$ and the final vertex is $\PP^{N(k_1,\cdots,k_n)}$. 
Note that the $(n-1)$ final edges of the hypercube (those edges whose target is the final vertex $\PP^{N(k_1,\cdots,k_n)}$)
are given by Segre embeddings of bipartite-type
\[
f_{N(k_1,\cdots,k_j),N(k_{j+1},\cdots, k_n)}:\PP^{N(k_1,\cdots,k_j)}\times \PP^{N(k_{j+1},\cdots,k_n)}\lra 
\PP^{N(k_{1},\cdots,k_n)},
\]
where $1\leq j\leq n-1$.

\begin{exam}[4-particle states]\label{4qb}We have already seen the examples of two and three particle states in the warm-up section. For the generalized Segre embedding $f_{1,1,1,1}$ characterizing entanglement of four particle states of spin-${1\over 2}$, we obtain a cube with commutative faces
\SelectTips{eu}{12}
\[ \xymatrix{ \PP^1\times\PP^1\times\PP^1\times\PP^1 \ar[dd]_{f_{1,1}\times\mathbb{I}\times\mathbb{I}}\ar[rd]^{\mathbb{I}\times f_{1,1}\times\mathbb{I}} \ar[rr]^{\mathbb{I}\times\mathbb{I}\times f_{1,1}} && 
\PP^1\times\PP^1\times\PP^3 \ar'[d][dd]^{f_{1,1}\times\mathbb{I}} \ar[rd]^{\mathbb{I}\times f_{1,3}} \\
& \PP^1\times\PP^3\times\PP^1 \ar[dd]_(.3){f_{1,3}\times\mathbb{I}} \ar[rr]^(.4){\mathbb{I}\times f_{3,1}} && \PP^1\times\PP^7 \ar[dd]^{f_{1,7}} \\
\PP^3\times\PP^1\times\PP^1 \ar'[r][rr]^(.2){\mathbb{I}\times f_{1,1}} \ar[rd]^{f_{3,1}\times\mathbb{I}}  && \PP^3\times\PP^3 \ar[rd]^{f_{3,3}} \\
& \PP^7\times\PP^1 \ar[rr]^{f_{7,1}} && \PP^{15} }.
\]
\end{exam}

We next state a general decomposability result for arbitrary products of projective spaces.
Since our interest lies in spin-$\frac{1}{2}$ particle systems, for the sake of simplicity we will restrict to 
the case where the initial spaces are projective lines.
For any integer $m\geq 1$, let
\[N_m:=N(1,\stackrel{(m)}{\cdots},1)=2^m-1.\]
We will consider the decompositions associated to the generalized Segre embedding 
\[\PP^1\times\stackrel{(n)}\cdots\times \PP^1\lra \PP^{N_n}.\]

\begin{defi}\label{defdecompo}Let $n\geq 2$ and $1<q\leq n$ be integers.
We will say that a point $z\in\PP^{N_n}$ is \textit{$q$-decomposable}
if and only if there exist positive integers
$m_1,\cdots,m_q$ with $m_1+\cdots+m_q=n$ such that 
\[z\in \Sigma_{N_{m_1},\cdots,N_{m_q}}.\]
Note that $q$-decomposable implies $(q-1)$-decomposable.
If $z$ is not $2$-decomposable, we will say that it is \textit{indecomposable}.
\end{defi}

Points that are $q$-decomposable will correspond precisely
to $q$-partite states and indecomposable points will correspond to entangled states. 

\begin{exam}[2- and 3-particle states]
A point $z$ in $\PP^3$ is
$2$-decomposable if and only if $z\in \Sigma_{1,1}$.
Otherwise it is indecomposable.
A point $z$ in $\PP^7$ is 
$2$-decomposable if and only if 
 $z\in \Sigma_{3,1}\cup \Sigma_{1,3}$. It is
$3$-decomposable if and only if
 $z\in \Sigma_{1,1,1}$.
The following result shows that $z$ is actually $3$-decomposable if and only if 
$z\in\Sigma_{3,1}\cap\Sigma_{1,3}$, so that it is $2$-decomposable in every possible way. In physical terms, it just says that a state is $3$-partite if and only if it is $2$-partite when considering both types of bipartitions.
\end{exam}

\begin{theo}[Generalized Decomposability]\label{entanglement_Segre}
Let $n\geq 2$ and $1<q\leq n$ be integers.
A point $z\in \PP^{N_n}$ is $q$-decomposable if and only if 
it lies in at least $q-1$ different Segre varieties of the form 
$\Sigma_{N_\ell,N_{n-\ell}}$, with $1\leq \ell\leq n-1$.
\end{theo}

For the particular extreme cases we have that a point $z\in\PP^{N_n}$ is:
\[\left\{\def\arraystretch{1.4}
\begin{array}{lll}
 \text{Indecomposable}&\Longleftrightarrow& z\notin \Sigma_{N_\ell,N_{n-\ell}},\text{ for all }1\leq \ell\leq n-1.\\
\text{$n$-decomposable}&\Longleftrightarrow& z\in \Sigma_{N_\ell,N_{n-\ell}},\text{ for all }1\leq \ell\leq n-1.\\
\end{array}\right.
\]

We refer to the Appendix for the proof.
This result will be essential in the next section, where we give a general method for measuring decomposability. Indeed, Theorem \ref{entanglement_Segre} asserts that decomposability is entirely determined by the Segre varieties 
$\Sigma_{N_\ell,N_{n-\ell}}$
for all $1\leq \ell\leq n-1$. Note that this family of varieties is 
the one arising when looking at the $(n-1)$ edges whose target is the last vertex of the $(n-1)$-dimensional hypercube and corresponds precisely to the family of all possible bipartitions of $\PP^{N_n}$.

This result will translate into taking $(n-1)$ measures of a given $n$-particle state, in order to detect the level of decomposability of its associated point in the projective space.

\begin{exam}[4-particle states]
In the situation of Example \ref{4qb} and in view of Theorem \ref{entanglement_Segre}, a point $z$ in $\PP^{15}$ is:
\[\left\{\def\arraystretch{1.4}
\begin{array}{lll}
 \text{indecomposable}&\Longleftrightarrow& z\notin \Sigma_{7,1}\cup \Sigma_{1,7}\cup \Sigma_{3,3}\\
 \text{$2$-decomposable}&\Longleftrightarrow& z\in \Sigma_{7,1}\cup \Sigma_{1,7}\cup \Sigma_{3,3}\\
 \text{$3$-decomposable}&\Longleftrightarrow&z\in (\Sigma_{1,7}\cap \Sigma_{7,1})\cup (\Sigma_{1,7}\cap\Sigma_{3,3})
 \cup(\Sigma_{7,1}\cap \Sigma_{3,3}).\\
 \text{$4$-decomposable}&\Longleftrightarrow&z\in  \Sigma_{7,1}\cap \Sigma_{1,7}\cap \Sigma_{3,3}
\end{array}\right.
\]
\end{exam}

\section{Operators controlling edges of the hypercube}\label{SecOps}

Given an $n$-particle state, in this section we define $(n-1)$ observables which measure its entanglement.
Let us first fix some notation.
We will denote by 
\[\sigma_0=\left(\begin{matrix}
                  1&0\\0&1
                 \end{matrix}
\right)\,;\,
\sigma_1=\left(\begin{matrix}
                  0&1\\1&0
                 \end{matrix}
\right)\,;\,
\sigma_2=\left(\begin{matrix}
                  0&-i\\i&0
                 \end{matrix}
\right)\,;\,
\sigma_3=\left(\begin{matrix}
                  1&0\\0&-1
                 \end{matrix}
\right)
\] 
and by $\mathbb{I}_k$ the identity matrix of size $k$.

The \textit{Kronecker product} of two matrices $\Aa=(a_{ij})\in \mathrm{Mat}_{k\times k}$ and $\Bb \in \mathrm{Mat}_{n\times n}$ is
the matrix of size $kn\times kn$ given by:
\[\Aa\otimes \Bb:=\left(
\begin{matrix}
a_{11}\Bb&\cdots& a_{1k}\Bb\\
\vdots&\vdots&\vdots\\
a_{k1}\Bb&\cdots&a_{kk}\Bb
\end{matrix}
\right).\]

The \textit{Hermitian product} of two complex 
 vectors $u=(u_0,\cdots,u_k)$ and $v=(v_0,\cdots,v_k)$ is
\[u\cdot v=\sum_{i=0}^k \overline{u_i}{v_i}.\]
Also, let 
\[||u||^2:=u\cdot u=\sum_{i=0}^k \overline{u}_i{u_i}.\]
For $\alpha$ a complex number, we denote $|\alpha|:=||\alpha||=\sqrt{\overline{\alpha}\cdot{\alpha}}$ its absolute value.

We will use the \textit{Lagrange identity}, which states that
\[||u||^2\cdot ||v||^2-|u\cdot v|^2=\sum_{i=0}^{k-1}\sum_{j=i+1}^k |u_iv_j-u_jv_i|^2.\]
Note that in many references, the term on the right side of the identity is often written in the equivalent form 
$|u_i\overline{v}_j-u_j\overline{v}_i|^2$ instead of $|u_iv_j-u_jv_i|^2$.

In order to describe $n$-particle states we fix an ordered basis
of $N_n=2^n-1$ linearly independent vectors of the corresponding Hilbert space
\[\ket{i_1\cdots i_n}=\ket{i_1}_{\Oo_1}\otimes\cdots\otimes\ket{i_n}_{\Oo_n}\]
where $\Oo_1,\cdots,\Oo_n$ denote the different observers and
$\{i_1,\cdots,i_n\}\in\{0,1\}$, since we are in the spin-$\frac{1}{2}$ case.
Using this basis, the coordinates for a general pure state for $n$ particles will be written as
\[\ket{\psi}_{\Oo_1\cdots\Oo_n}=\left(\begin{matrix}z_0\\z_1\\\vdots\\z_{N_n}\end{matrix}\right),\]
where $z_{i}$ are complex numbers satisfying the normalization condition
\[\sum_{i\geq 0} |z_i|^2=1.\]
Likewise, we will write:
\[\bra{\psi}=(\overline{z}_0,\overline{z}_1,\cdots,\overline{z}_{N_n}).\]

Given such a state, we may consider its class $[\psi]\in \PP^{N_n}$ by taking its homogeneous coordinates 
\[[\psi]=[z_0:\cdots:z_{N_n}],\]
where we recall that $[z_0:\cdots:z_{N_n}]=[\lambda z_0:\cdots: \lambda z_{N_n}]$ for any $\lambda\in\CC^*$.

For all $\ell=1,\cdots,n-1$, define the following observable acting on $n$-particle states:
\[\Jj_{n,\ell}(\psi):=2-\left({1\over 2^{\ell-1}}\sum_{i_1,\cdots,i_\ell=0}^3
|\bra{\psi}\sigma_{i_1}\otimes\cdots\otimes\sigma_{i_\ell}\otimes \mathbb{I}_{2^{n-\ell}}\ket{\psi}|^2\right).
\]

The purpose of this section is to show that
$\Jj_{n,\ell}(\psi)=0$ if and only if the class $[\psi]$ in the projective space $\PP^{N_n}$ lies in the Segre variety $\Sigma_{N_\ell,N_{n-\ell}}$.

The next two examples detail the computations of the observables introduced in Section \ref{Sec2part}
for the cases of two and three particles respectively.
Note that we used a slightly different notation, namely:
\begin{itemize}
 \item []2 particles: $\Jj_{A\otimes B}\equiv\Jj_{2,1}$ and $\Sigma_{A\otimes B}\equiv \Sigma_{1,1}$.
 \item []3 particles: $\Jj_{A\otimes BC}\equiv\Jj_{3,1}$, $\Sigma_{A\otimes BC}\equiv \Sigma_{1,3}$, 
$\Jj_{AB\otimes C}\equiv \Jj_{3,2}$,
and $\Sigma_{AB\otimes C}\equiv \Sigma_{3,1}$.
\end{itemize}

\begin{exam}[2-particle states]We have a single observable
\[\Jj_{2,1}(\psi)=2-\left(\sum_{i=0}^3 |\bra{\psi}\sigma_i\otimes \mathbb{I}_2\ket{\psi}|^2\right).\]
Define vectors $\Aa_0=(z_0,z_1)$ and $\Aa_1=(z_2,z_3)$, so that 
$\bra{\psi}=(\Aa_0,\Aa_1)$. Then we have
\[\def\arraystretch{1.6}
\begin{array}{ll}
\Jj_{2,1}(\psi)&=2-\left(|\Aa_0{\cdot}\Aa_0+\Aa_1{\cdot}\Aa_1|^2+|\Aa_0{\cdot}\Aa_1+\Aa_1{\cdot}\Aa_0|^2\right.\\
&\left.+|\Aa_0{\cdot}\Aa_1-\Aa_1{\cdot}\Aa_0|^2+|\Aa_0{\cdot}\Aa_0-\Aa_1{\cdot}\Aa_1|^2\right)=
\\
&
=1-\left(\left|\ov{z}_0z_0+\ov{z}_1z_1+\ov{z}_2z_2+\ov{z}_3z_3\right|^2+
\left|\ov{z}_0z_2+\ov{z}_1z_3+\ov{z}_2z_0+\ov{z}_3z_1\right|^2+\right.\\
&+\left.\left|-\ov{z}_0z_2-\ov{z}_1z_3+\ov{z}_2z_0+\ov{z}_3z_1\right|^2+
\left|\ov{z}_0z_0+\ov{z}_1z_1-\ov{z}_2z_2-\ov{z}_3z_3\right|^2
\right)=\\
&=1-\left(\left|\ov{z}_0z_2+\ov{z}_1z_3+\ov{z}_2z_0+\ov{z}_3z_1\right|^2+
\left|-\ov{z}_0z_2-\ov{z}_1z_3+\ov{z}_2z_0+\ov{z}_3z_1\right|^2+\right.\\
&\left.+\left|\ov{z}_0z_0+\ov{z}_1z_1-\ov{z}_2z_2-\ov{z}_3z_3\right|^2\right)=
4\left|z_0z_3-z_1z_2\right|^2.
\end{array}\]
Therefore $\Jj_{2,1}(\psi)=0$ if and only if $[\psi]\in \Sigma_{1,1}$.
\end{exam}

\begin{exam}[3-particle states]In this case we have two observables
\[\def\arraystretch{1.6}
\begin{array}{l}
\Jj_{3,1}(\psi):=2-\left(\sum_{i=0}^3 |\bra{\psi}\sigma_i\otimes I_4\ket{\psi}|^2\right)
\text{ and }\\
\Jj_{3,2}(\psi):=2-\left({1\over 2}\sum_{i,j=0}^3 |\bra{\psi}\sigma_i\otimes \sigma_j\otimes I_2\ket{\psi}|^2\right).
\end{array}\]
Write $z=(\Aa_0,\Aa_1)$ where $\Aa_0=(z_0,\cdots,z_3)$ and $\Aa_1=(z_4,\cdots,z_7)$.
Then we have 

\[\def\arraystretch{1.6}
\begin{array}{ll}
\Jj_{3,1}(\psi)=&2-(|\ov{\Aa}_0\Aa_0+\ov{\Aa}_1 \Aa_1|^2+ |\ov{\Aa}_0\Aa_0+\ov{\Aa}_1 \Aa_1|^2+
|\ov{\Aa}_0 \Aa_1+\ov{\Aa}_1 \Aa_0|^2+\\
&+|-\ov{\Aa}_0 \Aa_1+\ov{\Aa}_1 \Aa_0|^2+
|\ov{\Aa}_0 \Aa_0-\ov{\Aa}_1 \Aa_1|^2)=\\
&=1-(|\ov{\Aa}_0 \Aa_1+\ov{\Aa}_1 \Aa_0|^2+
|-\ov{\Aa}_0 \Aa_1+\ov{\Aa}_1 \Aa_0|^2+
|\ov{\Aa}_0 \Aa_0-\ov{\Aa}_1 \Aa_1|^2)=
\\
&=4\left\{
|z_0z_5-z_1z_4|^2+|z_0z_6-z_2z_4|^2+|z_0z_7-z_3z_4|^2+\right.\\
&\left.+|z_1z_6-z_2z_5|^2+|z_1z_7-z_3z_5|^2+|z_2z_7-z_3z_6|^2
\right\}.\,
\end{array}
\]

Note that the numbers $z_iz_j-z_j'z_i'$ correspond to 
the minors describing the zero locus of $\Sigma_{1,3}$. Indeed, this is determined by the vanishing of all the $2\times 2$ minors of the matrix 
\[
\left(
\begin{matrix}
z_0&z_1&z_2&z_3\\
z_4&z_5&z_6&z_7
\end{matrix}
\right).
\]
Therefore $\Jj_{3,1}(\psi)=0$ if and only if $[\psi]\in \Sigma_{1,3}$.

Likewise, writing $z=(\Aa_0,\Aa_1,\Aa_2,\Aa_3)$ with $\Aa_0=(z_0,z_1)$, $\Aa_1=(z_2,z_3)$, 
$\Aa_2=(z_4,z_5)$ and $\Aa_3=(z_6,z_7)$ we easily obtain
\begin{align*}
\Jj_{3,2}(\psi)=4\left\{
|z_0z_3-z_1z_2|^2+|z_0z_5-z_1z_4|^2+|z_0z_7-z_1z_6|^2+\right.\\
\left.+|z_2z_5-z_3z_4|^2+|z_2z_7-z_3z_6|^2+|z_4z_7-z_5z_6|^2
\right\}.\,\end{align*}
Note that the numbers $z_iz_j-z_j'z_i'$ correspond to 
the minors describing the zero locus of $\Sigma_{3,1}$. Indeed, this is determined by the vanishing of all the $2\times 2$ minors of the matrix 
\[
\left(
\begin{matrix}
z_0&z_1\\
z_2&z_3\\
z_4&z_5\\
z_6&z_7
\end{matrix}
\right).
\]
Therefore $\Jj_{3,2}(\psi)=0$ if and only if $[\psi]\in \Sigma_{3,1}$.
\end{exam}

Returning to the general setting, note that the measures
$\Jj_{n,\ell}(\psi)$ are related to
the trace of the squared density matrix for a given partition of the system.
Indeed, the density matrix for any physical state is a Hermitian operator and
therefore can be written in terms of $\sigma_{i}$ as
\[\rho = \sum_{i_1, ... i_n\in\{1,2,3\}}
c_{i_1,\cdots i_l}\sigma_{i_1}\otimes\cdots\otimes\sigma_{i_l}\otimes\mathbb{I}_{2^{n-\ell}}
+ d_{i_{l+1}\cdots i_n}
\mathbb{I}_{2^{\ell}}\otimes\sigma_{i_{l+1}}\otimes\cdots\otimes\sigma_{i_n},\]
where we write explicitly the partition of the system in $A$ and $B$.
In this notation, the reduced density matrix of the system $A$ reads
\[\rho_A = 2^{n-l}c_{i_1,\cdots i_l}\sigma_{i_1}\otimes\cdots\otimes\sigma_{i_l}\].
We may now compute expectation values for our
operators:
\[
\bra{\psi}\sigma_{i_1}\otimes\cdots\otimes\sigma_{i_\ell}\otimes
\mathbb{I}_{2^{n-\ell}}\ket{\psi} =
{\rm Tr}\left(\rho\cdot(\sigma_{i_1}\otimes\cdots\otimes\sigma_{i_l}\otimes\mathbb{I}_{2^{n-\ell}})\right),\]
 and using the identity
 \[{\rm Tr} \left((\sigma_{i_1}\otimes\cdots\otimes\sigma_{i_n})\cdot(
\sigma_{j_1}\otimes\cdots\otimes\sigma_{j_n})\right) = 2^n \delta_{i_1
  j_1}\delta_{i_2 j_2}\cdots \delta_{i_n j_n}\] we obtain
\[c_{i_1,\cdots i_l} = 2^n \bra{\psi}\sigma_{i_1}\otimes\cdots\otimes\sigma_{i_\ell}\otimes
\mathbb{I}_{2^{n-\ell}}\ket{\psi}.\]
We now compute the trace of the squared density matrix:
\[{\rm Tr}_A \rho_A^2 = 2^{2(n-l)} c_{i_1,\cdots i_l} c_{j_1,\cdots j_l} {\rm Tr} \left((\sigma_{i_1}\otimes\cdots\otimes\sigma_{i_n})\cdot(
\sigma_{j_1}\otimes\cdots\otimes\sigma_{j_n})\right) = 2^l \bra{\psi}\sigma_{i_1}\otimes\cdots\otimes\sigma_{i_\ell}\otimes
\mathbb{I}_{2^{n-\ell}}\ket{\psi}^2\] 
which leads to the identity
\[{\cal J}_{n,\ell}(\psi)=2 \left(1 - {\rm Tr_A} \rho_A^2 \right).\]

The main result of this section is the following:

\begin{theo}\label{teomeasure}
 Let $\ket{\psi}$ be a pure $n$-particle state and let $1\leq \ell\leq n-1$.
 Then 
\[\Jj_{n,\ell}(\psi)=4\sum_I |M_{I}|^2,\]
where the sum runs over all $2\times 2$ minors $M_I$ determining the zero locus of $\Sigma_{N_\ell,N_{n-\ell}}$.
In particular,
\[\Jj_{n,\ell}(\psi)=0\Longleftrightarrow [\psi]\in \Sigma_{N_\ell,N_{n-\ell}}.\]
\end{theo}
\begin{proof}Write $(z_0,\cdots,z_{N_n})=(\Aa_0,\cdots,\Aa_{N_\ell})$, where 
\[\Aa_j=(z_{j(N_{n-\ell}+1)},\cdots,z_{j(N_{n-\ell}+1)+N_{n-\ell}})\]
for $j=0,\cdots,N_\ell$,
so that each $\Aa_j$ is tuple with ${N_{n-\ell}+1}$ components.
With this notation, we have 
\[\Jj_{n,\ell}(\psi)=4\sum_{j=0}^{N_\ell-1}\sum_{k>j}^{N_\ell}( ||\Aa_j||^2\cdot ||\Aa_k||^2-|\Aa_j\cdot \Aa_k|^2).\]
Applying the Lagrange identity we obtain
\[\Jj_{n,\ell}(\psi)=4\sum_{j=0}^{N_\ell-1}\sum_{k>j}^{N_\ell}
\sum_{s=0}^{N_\ell-1}\sum_{t>s}^{N_\ell} 
|\Aa_{j,s}\cdot \Aa_{k,t}-\Aa_{j,t}\cdot \Aa_{k,s}|^2
,\]
where $\Aa_{i,j}$ denotes the $j$-th component of the tuple $\Aa_i$.
It now suffices to note that the numbers 
\[\Aa_{j,s}\Aa_{k,t}-\Aa_{j,t}\Aa_{k,s}\]
correspond to the minors $M_I$. Indeed, the zero locus of $\Sigma_{N_\ell,N_{n-\ell}}$ is determined by the vanishing of the $2\times 2$ minors of the matrix 
\[
 \left(
\begin{matrix}
 \Aa_{0,0}&\cdots&\Aa_{0,N_{n-\ell}}\\
 \vdots&&\vdots\\
  \Aa_{N_\ell,0}&\cdots&\Aa_{N_\ell,N_{n-\ell}}
\end{matrix}
 \right).\qedhere
\]
\end{proof}

\section{Physical interpretation}\label{Secphysical}
Given an integer $n\geq 1$, we fix an ordered basis of 
 $N_n=2^n-1$ linearly independent vectors of the Hilbert space
of pure $n$-particle states of spin-$\frac{1}{2}$
 \[\ket{i_1\cdots i_n}_{\Oo_1\cdots\Oo_n}=\ket{i_1}_{\Oo_1}\otimes\cdots\otimes\ket{i_n}_{\Oo_n},\]
where $\Oo_1,\cdots,\Oo_n$ denote the different observers.
We will omit the labels of the observers, but one should note that, in the following, the chosen basis always has the same fixed order unless stated otherwise.

\begin{defi}Let $1\leq q\leq n$ be an integer.
An $n$-particle state $\ket{\psi}$ is said to be \textit{$q$-partite} if it can be written as  
\[\ket{\psi}=\ket{\psi_{1}}\otimes\cdots\otimes\ket{\psi_{q}},\]
where $\ket{\psi_i}$ are $n_i$-particle states, with $n_i>0$ and $n_1+\cdots+n_q=n$.
\end{defi}

$1$-partite states are called \textit{entangled}, 
while $n$-partite states are called \textit{separable}.
Physically, separable states are those that are uncorrelated. A product state can thus be easily
prepared in a local way: each observer $\Oo_i$ produces the state $\ket{\psi_i}$ and the measurement
outcomes for each observer do not depend on the outcomes for the other observers.

A basic observation is that a state is separable if and only if it lies in the generalized Segre variety of Definition \ref{defgenseg}
(see for instance \cite{Ballico}). Likewise, a state is $q$-partite if and only if its corresponding projective point on $\PP^{N_n}$ lies in a Segre variety of the form 
\[\Sigma_{N_{m_1},\cdots,N_{m_q}}\]
with $m_1,\cdots,m_q$ positive integers such that $m_1+\cdots+m_q=n$. So $q$-partite states correspond geometrically to the $q$-decomposable points of Definition \ref{defdecompo}.
Combining the results of the previous two sections we have that an
$n$-particle state $\ket{\psi}$ is $q$-partite if and only if 
there are $q-1$ indices $\ell_1,\cdots,\ell_{q-1}$ with $1\leq \ell_i\leq n-1$
and $\ell_i\neq \ell_j$
such that $\Jj_{n,\ell_i}(\psi)=0$.
Indeed, from Theorem \ref{entanglement_Segre}
we know that $\ket{\psi}$ is $q$-partite
if and only if its corresponding point $[\psi]$ in $\PP^{N_n}$
lies in at least $q-1$ different Segre varieties of the form 
$\Sigma_{N_\ell,N_{n-\ell}}$, with $1\leq \ell\leq n-1$.
Moreover, from Theorem \ref{teomeasure} we know that $[\psi]\in \Sigma_{N_\ell,N_{n-\ell}}$ if and only if $\Jj_{n,\ell}(\psi)=0$.

Note that, a priori, given an $n$-particle state, one would have to take $(n-1)!$ measures of bipartite type to completely determine its 
$q$-decomposability (namely, for each possible bipartition, check further bipartitions recursively). The hypercube approach tells us that it suffices to take $(n-1)$ measures, corresponding to the operators $\{\Jj_{n,\ell}\}$ in order to determine completely its $q$-decomposability. 

In the remaining of the section, we study the behaviour of the observables $\Jj_{n,\ell}$ in some particular cases of interest. We first introduce the average observable acting on $n$-particle states:
\[\Jj(\psi):={1\over n-1}\sum_{\ell=1}^{n-1}\Jj_{n,\ell}(\psi).\]

The two- and three- particle states discussed in Section \ref{Sec2part} which are invariant under permutations of the Hilbert basis ($\ket{\mathrm{Sep}}$, $\ket{\mathrm{EPS}}$, $\ket{\mathrm{GHS}}$, $\ket{W}$) allow for natural generalizations to the $n$-particle case. We study their entanglement measures.

Note first that the separable $n$-particle state
\[\ket{\mathrm{Sep}_n}:=\ket{0 \,\stackrel{(n)}{\cdots}\, 0}\]
corresponds to the point in $\PP^{N_n}$ given by 
\[[\mathrm{Sep}_n]=[1:0: \,{ \cdots } \,:0]\]
and so one easily verifies that $\Jj(\mathrm{Sep}_n)=0$. 

The \textit{Schrödinger $n$-particle state}
is a superposition of two maximally distinct states
\[\ket{S_n}:={1\over \sqrt{2}}\left(\ket{0\,\stackrel{(n)}{\cdots}\, 0}+\ket{1\,\stackrel{(n)}{\cdots} \,1}\right).\]
It generalizes the two-particle state $\ket{\mathrm{EPS}}$ and the three-particle state $\ket{\mathrm{GHZ}}$ and it
corresponds to the point in $\PP^{N_n}$ given by 
\[[S_n]=[1:0:\,\cdots\,:0:1].\]
One easily computes $\Jj_{n,\ell}(S_n)=1$ for all $1\leq \ell\leq n$ and so $\Jj(S_n)=1$.
For $n>3$, its not clear that the Schrödinger state exhibits maximal entanglement \cite{Higuchi}.
This observation agrees with our measures for $\Jj_{n,\ell}$, as we will see below.

We now consider a generalization of the $W$ state for three particles, to the case of $n$-particles. For each fixed integer $0\leq k\leq n$, these states are constructed by adding all permutations of generators of the form
\[\ket{1}\otimes\stackrel{(k)}{\cdots}\otimes\ket{1}\otimes \ket{0}\otimes\stackrel{(n-k)}{\cdots}\otimes\ket{0}\]
with $k$ states $\ket{1}$ and $n-k$ states $\ket{0}$, together with a global normalization constant. These are clearly invariant with respect to permutations of the basis. Denote such states by 
\[\ket{D_{n,k}}={\binom{n}{k}}^{-{1\over 2}}\sum_{\mathrm{permut}} \ket{1}\otimes\stackrel{(k)}{\cdots}\otimes\ket{1}\otimes \ket{0}\otimes\stackrel{(n-k)}{\cdots}\otimes\ket{0}.\]
These are known as \textit{Dicke states} \cite{Dicke}.
Note that $\ket{D_{n,0}}=\ket{\mathrm{Sep}_n}$ and $\ket{D_{3,1}}=\ket{W_3}$. In the case of four particles we have 
\[\def\arraystretch{1.6}
\begin{array}{ll}
\ket{D_{4,1}}={1\over \sqrt{4}}\left(\ket{1000}+\ket{0100}+\ket{0010}+\ket{0001}\right),\\
\ket{D_{4,2}}={1\over \sqrt{6}}\left(\ket{1100}+\ket{1010}+\ket{1001}+\ket{0110}
+\ket{0101}+\ket{0011}\right),\\
\ket{D_{4,3}}={1\over \sqrt{4}}\left(\ket{1110}+\ket{1101}+\ket{1011}+\ket{0111}\right).
\end{array}
\]
Their corresponding points in $\PP^{15}$ are 
\[\def\arraystretch{1.6}
\begin{array}{ll}
\ket{D_{4,1}}=[0:1:1:0:1:0:0:0:1:0:0:0:0:0:0:0],\\
\ket{D_{4,2}}=[0:0:0:1:0:1:1:0:0:1:1:0:1:0:0:0],\\
\ket{D_{4,3}}=[0:0:0:0:0:0:0:1:0:0:0:1:0:1:1:0].
\end{array}
\]
We obtain the following table for the observables $\Jj_{4,\ell}$:
\[
\arraycolsep=4pt\def\arraystretch{1.4}
 \begin{array}{|c|c|c|c|c|c|c|c|}
 \hline
 & \Ket{D_{4,1} }          & \Ket{D_{4,2}} & \Ket{D_{4,3}}          \\
   \hline\hline
   \Jj_{4,1}  & \frac{3}{4} & 1 & \frac{3}{4}  \\
   \Jj_{4,2} & 1           & 1 & 1            \\
   \Jj_{4,3}  & \frac{3}{4} & 1 & \frac{3}{4}  \\
 \hline
 \Jj&{5\over 6}&1&{5\over 6}\\
 \hline
 \end{array}\]
 
 Note $\Jj(D_{4,2})=1$ as is the case for the state $\ket{S_4}$.
For more than four particles we obtain values $>1$ for the observables $\Jj_{n,\ell}$. For instance, in the five particle case, we have:

\[
\arraycolsep=4pt\def\arraystretch{1.4}
 \begin{array}{|c|c|c|c|c|c|c|c|}
 \hline
  &  \Ket{D_{5,1}}          &  \Ket{D_{5,2}}             &  \Ket{D_{5,3}}            &  \Ket{D_{5,4}}             \\
   \hline\hline
   \Jj_{5,1}  & \frac{16}{25}& \frac{24}{25}  & \frac{24}{25}  & \frac{16}{25}  \\
   \Jj_{5,2}  & \frac{24}{25}& \frac{27}{25} & \frac{27}{25} & \frac{24}{25}  \\
   \Jj_{5,3}  & \frac{24}{25}& \frac{27}{25} & \frac{27}{25} & \frac{24}{25}  \\
   \Jj_{5,4}  & \frac{16}{25}& \frac{24}{25}  & \frac{24}{25}  & \frac{16}{25}  \\
   \hline
   \Jj&{4\over 5}&{51\over 50}&{51\over 50}&{4\over 5}\\
 \hline
 \end{array}.\]
 In particular, we see that 
 \[\Jj(D_{5,2})=\Jj(D_{5,3})>\Jj(S_5)=1.\]
 For higher particle states the same pattern repeats itself, with the middle states 
 \[\ket{D_{n,\lfloor{n\over 2}\rfloor}}=\ket{D_{n,\lceil{n\over 2}\rceil}}\] always exhibiting the largest entanglement as well as symmetries of the tables in both directions.

We end this section with some notable examples in the four- and five-particle cases.
The state 
\[\ket{\mathrm{HS}}={1\over \sqrt{6}}\left(
\ket{1100}+\ket{0011}+\omega \ket{1001}+ \omega \ket{0110}+ \omega^2\ket{1010}+\omega^2\ket{0101}
\right)\]
where $\omega=e^{2\pi i\over 3}$ is a third root of unity,
was conjectured to be
maximally entangled by Higuchi-Sudbery \cite{Higuchi} and it actually
gives a local maximum of the average two-particle von
Neumann entanglement entropy \cite{Brierley}.
Another highly (though not maximally) entangled state, found by 
Brown-Stepney-Sudbery-Braunstein
\cite{Brown}, is
given by 
\[\ket{\mathrm{BSSB}_4}={1\over 2}\left(
\ket{0000}+\ket{+}\otimes\ket{011}+\ket{1101}+\ket{-}\otimes\ket{110}
\right),\]
where $\ket{+}={1\over \sqrt{2}}(\ket{0}+\ket{1})$ and $\ket{-}={1\over \sqrt{2}}(\ket{0}-\ket{1})$.
Our measures agree with these facts, as shown in the table below.
\[
\arraycolsep=4pt\def\arraystretch{1.4}
 \begin{array}{|c|c|c|c|c|c|c|c|}
 \hline
           &  \Ket{S_4}              &  \Ket{D_{4,2}}            & \ket{\mathrm{BSSB}_4}  &\Ket{\mathrm{HS}}             \\
   \hline\hline
   \Jj_{4,1}  &  1&  1&{3\over 4}&1\\
   \Jj_{4,2}   &  1&  1&{5\over 4}&{4\over 3}\\
   \Jj_{4,3}   &  1&  1&1&1\\
   \hline
   \Jj&1&  1&  1&{10\over 9}\\
 \hline
 \end{array}.\]
In \cite{Xinos}, a related measure of entanglement is introduced,
based on vector lengths and the angles between vectors of certain coefficient
matrices. While this measure is strongly related to concurrence and hence to the observables $\Jj$, their measure $E_{avg}$ does not distinguish the states  $\ket{\mathrm{BSSB}_4}$ and $\Ket{\mathrm{HS}}$. In contrast, we do find that 
\[1=\Jj(\mathrm{BSSB}_4)<\Jj(\mathrm{HS}).\]
In \cite{Brown}, a highly entangled five-particle state is described as 
\[\ket{\mathrm{BSSB}}_5={1\over 2}\left(
\ket{000}\otimes\ket{\Phi_-}+
\ket{010}\otimes\ket{\Psi_-}+
\ket{100}\otimes\ket{\Phi_+}+
\ket{111}\otimes\ket{\Psi_+}
\right),\]
 where $\ket{\Psi_{\pm}}=\ket{00}\pm\ket{11}$ and $\ket{\Phi_{\pm}}=\ket{01}\pm\ket{10}$.
 Our measures give:
\[
\arraycolsep=4pt\def\arraystretch{1.4}
 \begin{array}{|c|c|c|c|c|c|c|c|}
 \hline
  &  \Ket{S_5}          &  \Ket{D_{5,2}}             &  \Ket{\mathrm{BSSB}_5}           \\
   \hline\hline
   \Jj_{5,1}  & 1 &\frac{24}{25}  &  1  \\
   \Jj_{5,2}  & 1&\frac{27}{25} &  {3\over 2}  \\
   \Jj_{5,3}  & 1& \frac{27}{25} &  {5\over 4}  \\
   \Jj_{5,4}  & 1& \frac{24}{25}  &  1  \\
   \hline
   \Jj&1&{51\over 50}&{19\over 16}\\
 \hline
 \end{array}\]
In particular, we see that 
\[1=\Jj(S_5)<\Jj(D_{5,2})< \Jj(\mathrm{BSSB}_5).\]
 
\section{Summary and conclusions}

Within a purely geometric framework, we have described entanglement of $n$-particle states in terms of a hypercube of bipartite-type Segre maps.
For simplicity, in this paper we have restricted to spin-${1\over 2}$ particles, but the geometric results generalize almost verbatim to the case of qudits. The hypercube picture allows to identify 
separability (or more generally, $q$-decomposability) in terms of a depth factor within the hypercube: the deeper a state lies in the hypercube, the more separable.

We have defined a collection of operators which measure the properties of a state in the above geometric set-up.
Given an $n$-particle state and having fixed an ordered basis of the total Hilbert space, there are $2^{n-1}$ different decomposability possibilities, given by the different ordered $q$-partitions, for $1\leq q\leq n$. A standard way to characterize the decomposability of such a state would be to consider, for any possible bipartition, all of its possible bipartitions in a recursive way. This gives a total of $(n-1)!$ measures to be taken.
Our hypercube approach says that it suffices to take $(n-1)$ measures, given by the operators $\Jj_{n,\ell}$, with $1\leq \ell\leq n-1$. So as the complexity of the problem grows factorially, our solution just grows linearly on $n$.

The operators $\Jj_{n,\ell}$ measure the different bipartitions of
the system, corresponding geometrically with the last $(n-1)$ edges of
the Segre hypercube. The expected value of these operators is
related with the quantum
Tsalis entropy ($q=2$) of both parts of the state.
The concrete values of $\ell$ for which $\Jj_{n,\ell}$ vanishes shows
precisely in what edge of the hypercube the state belongs or,
more physically, in which parts the state is separable.  

To illustrate the physical interest of our approach, we have computed the value of the operators $\Jj_{n,\ell}$
for various entangled states considered in the literature. The motivation for many of them arises in quantum computing and are therefore classified from
the quantum control perspective. In all cases, the
proposed observables give results consistent with the expectations.

\newpage

\appendix

\section{Proof of the Generalized Decomposability Theorem}
This appendix is devoted to the proof of Theorem \ref{entanglement_Segre} on geometric decomposability.

Let us first consider the tripartite-type
Segre embedding
\[f_{k_1,k_2,k_3}
:\PP^{k_1}\times \PP^{k_2}\times\PP^{k_3}\to \PP^{N(k_1,k_2,k_3)},\]
where we recall that
\[N(k_1,\cdots,k_n):=(k_1+1)\cdots(k_n+1)-1.\]
The following lemma relates the Segre varieties associated to it.
\begin{lemm}\label{pullback}
Let $k_1,k_2,k_3$ be positive integers.
Then
\[\Sigma_{k_1,k_2,k_3}=\Sigma_{k_1,N(k_2,k_3)}\cap \Sigma_{N(k_1,k_2),k_3}.\]
\end{lemm}
\begin{proof}
It suffices to prove the inclusion $\Sigma_{k_1,N(k_2,k_3)}\cap \Sigma_{N(k_1,k_2),k_3}\subseteq \Sigma_{k_1,k_2,k_3}$.
Recall that we  have identities
\[f_{k_1,k_2,k_3}=
f_{k_1,N(k_2,k_3)}\circ(\mathbb{I}\times f_{k_2,k_3})=
f_{N(k_1,k_2),k_3}\circ(f_{k_1,k_2}\times\mathbb{I}).\]
Let $a=[a_i]$, $b=[b_j]$ and $c=[c_k]$ be coordinates for $\PP^{k_1}$, $\PP^{k_2}$ and $\PP^{k_3}$ respectively.
We will also let $x=[x_{ij}]$ and $y=[y_{jk}]$ be coordinates for $\PP^{N(k_1,k_2)}$ and $\PP^{N(k_2,k_3)}$ respectively,
and $z=[z_{ijk}]$ will denote coordinates for $\PP^{N(k_1,k_2,k_3)}$.

We have:
\begin{align*}
(f_{k_1,k_2}\times\mathbb{I})(a,b,c)=(x,c),\text{ with }x_{ij}:=a_ib_j.\\
(\mathbb{I}\times f_{k_2,k_3})(a,b,c)=(a,y),\text{ with }y_{jk}:=b_jc_k.\\
f_{N(k_1,k_2),k_3}(x,c)=(z),\text{ with }z_{ijk}:=x_{ij}c_k.\\
f_{k_1,N(k_2,k_3)}(a,y)=(z),\text{ with }z_{ijk}:=a_i y_{jk}.
\end{align*}
Given a point $z=(z_{ijk})\in \PP^{N(k_1,k_2,k_3)}$, we claim the following:
\begin{enumerate}
 \item $z\in\Sigma_{k_1,N(k_2,k_3)}$ if and only if
 all the $2\times 2$ minors of the matrix 
\[
\left(
\begin{matrix}
z_{000}&\cdots&z_{0 k_2 k_3}\\
\vdots&&\vdots\\
z_{k_100}&\cdots&z_{k_1 k_2 k_3}
\end{matrix}
\right)
\]
vanish, so that $z_{ijk}\cdot z_{i'j'k'}=z_{i'jk}\cdot z_{ij'k'}$ for all $i\neq i'$ and $(j,k)\neq (j',k')$.
\item $z\in\Sigma_{N(k_1,k_2),k_3}$ if and only if
all the $2\times 2$ minors of the matrix 
\[
\left(
\begin{matrix}
z_{000}&\cdots&z_{0 0 k_3}\\
\vdots&&\vdots\\
z_{k_1k_20}&\cdots&z_{k_1 k_2 k_3}
\end{matrix}
\right)
\]
vanish, so that $z_{ijk}\cdot z_{i'j'k'}=z_{ijk'}\cdot z_{i'j'k}$ for all $(i,j)\neq (i',j')$ and $k\neq k'$.
\item $z\in \Sigma_{k_1,k_2,k_3}$ if and only if 
$z\in\Sigma_{N(k_1,k_2),k_3}$, so that
$z_{ijk}=x_{ij} \cdot c_k$, and $x=(x_{ij})\in \Sigma_{k_1,k_2}$. This last condition gives the vanishing 
of the  $2\times 2$ minors of the matrix
\[\left(
\begin{matrix}
x_{00}&\cdots& x_{0k_2}\\
\vdots&&\vdots\\
x_{k_10}&\cdots&x_{k_1k_2}
\end{matrix}
\right).
\]
Therefore we have $x_{ij}\cdot x_{i'j'}=x_{ij'}\cdot x_{i'j}$ for all $i\neq i'$, $j\neq j'$. This gives identities
\[z_{ijk}\cdot z_{i'j'k'}=x_{ij}\cdot c_k\cdot x_{i'j'}\cdot c_{k'}=x_{i'j}\cdot c_k\cdot x_{ij'}\cdot c_{k'}
=z_{i'jk}\cdot z_{ij'k'}\]
for all $i\neq i'$, $j\neq j'$ and $k\neq k'$.
\end{enumerate}
Claims (1) and (2) are straightforward, while (3) follows from the identity
\[\Sigma_{k_1,k_2,k_3}=\mathrm{Im}(f_{N(k_1,k_2),k_3}\circ (f_{k_1,k_2}\times \mathbb{I})).\]

Assume now that $z\in\Sigma_{k_1,N(k_2,k_3)}\cap \Sigma_{N(k_1,k_2),k_3}$.
Then the equations in (1) and (2) are satisfied, and moreover we may write 
$z_{ijk}=x_{ij}\cdot c_k$. To show that 
$z\in \Sigma_{k_1,k_2,k_3}$ it
only remains to prove that 
$x_{ij}\cdot x_{i'j'}=x_{ij'}\cdot x_{i'j}$ for all $i\neq i'$ and $j\neq j'$. By (1) we have 
\[z_{ijk}\cdot z_{i'j'k'}=x_{ij}\cdot c_k\cdot x_{i'j'}\cdot c_{k'}=x_{i'j}\cdot c_k\cdot x_{ij'}\cdot c_{k'}
=z_{i'jk}\cdot z_{ij'k'}\]
for all $i\neq i'$ and $(j,k)\neq (j',k')$. Take $k=k'$ such that $c_k\neq 0$. This gives 
\[c_k^2\cdot x_{ij}\cdot x_{i'j'}=c_k^2\cdot x_{i'j}\cdot x_{ij'}\text{ for all }i\neq i'\text{ and }j\neq j'.\]
Since $c_k^2\neq 0$, we obtain the desired identities.
\end{proof}

In order to prove the Generalized Decomposability Theorem it will be useful to denote the cubical decomposition of 
\[\PP^1\times\stackrel{(n)}\cdots\times \PP^1\lra \PP^{N_n}\]
by the $(n-1)$-dimensional hypercube whose vertices are given by tuples 
$v=(v_1,\cdots,v_{n-1})$ where $v_i\in \{0,1\}$,
with initial vertex $v_0=(0,\cdots,0)$ representing $\PP^1\times\stackrel{(n)}{\cdots}\times\PP^1$
and final vertex $v_f=(1,\cdots,1)$
representing $\PP^{N_n}$.
We will call $|v|:=v_1+\cdots+v_{n-1}$ the \textit{degree of a vertex}. 
All edges are of the form 
\[(v_1,\cdots,v_{n-1})\stackrel{[j]}{\lra} (w_1,\cdots,w_{n-1})\] where $v_i=w_i$ for all $i\neq j$ 
and $w_j=v_j+1=1$, so that $|w|=|v|+1$.
Such an edge represents a contraction of a product $\times$ at the position $j$ via a Segre embedding.

\begin{exam}
For example, the cubical representation of
$f_{1,1,1}$ 
is
\[
\xymatrix{
(00)\ar[d]_{[2]}
\ar[rr]^{[1]}&&(10)\ar[d]^{[2]}\\
(01)\ar[rr]^{[1]}&&(11)
},
\]
and the cubical representation of
$f_{1,1,1,1}$ is
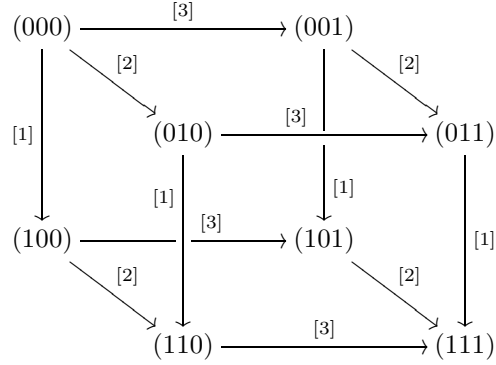
\SelectTips{eu}{12}
\[ \xymatrix{ (000) \ar[dd]_{[1]}\ar[rd]^{[2]} \ar[rr]^{[3]} && (001) \ar'[d][dd]^{[1]} \ar[rd]^{[2]} \\
& (010) \ar[dd]_(.3){[1]} \ar[rr]^(.4){[3]} && (011) \ar[dd]^{[1]} \\
(100) \ar'[r][rr]^(.2){[3]} \ar[rd]^{[2]}  && (101) \ar[rd]^{[2]} \\
& (110) \ar[rr]^{[3]} && (111) }.
\]
\end{exam}
We will say that a point $z\in\PP^{N_n}$ \textit{lives in a vertex} $v=(v_1,\cdots,v_{n-1})$
if it is in the image of the map $v\to v_f$ given by the composition of all edges $[i]$ such that $v_i=0$. 
It follows that $z$ is $q$-decomposable if and only if it lives in some vertex $v$ of degree $|v|=n-q$.

\begin{lemm}\label{pullback_living}
Let $v\neq v'$ be two vertices of the same degree in the $(n-1)$-dimensional hypercube. Then
there is a unique $2$-dimensional face of the form 
\[
\xymatrix{
\ar[d]_{[j]}w'\ar[r]^{[i]} &v'\ar[d]^{[j]}\\
v\ar[r]^{[i]}&w
}
\]
and if $z$ lives in $w$, $v$ and $v'$, then it also lives in $w'$.
\end{lemm}
\begin{proof}
Since $|v|=|v'|$ and $v\neq v'$, there are $i,j$ such that $v_i=0$, $v_i'=1$, $v_j=1$, $v_j'=0$,
and $v_k=v_k'$ for all $k\neq i,j$.
Let $w$ and $w'$ be the vertices whose components are given by
\[w_k=\mathrm{max} \{v_k,v_k'\}\text{ and }w_k'=\mathrm{min} \{v_k,v_k'\}.\]
This gives the above commutative square.
Assume now that $z$ lives in $w$, $v$ and $v'$. It suffices to consider two cases:

\textit{Case} $j=i+1$. In this case, the above commutative square represents morphisms of the form 
\[
\xymatrix{\ar[d]^{\mathbb{I}_A\times \mathbb{I}_{k_1}\times f_{k_2,k_3}\times \mathbb{I}_B}
A\times \PP^{k_1}\times \PP^{k_2}\times \PP^{k_3}\times B\ar[rrr]^-{\mathbb{I}_A\times  f_{k_1,k_2}\times \mathbb{I}_{k_3}\times \mathbb{I}_B} &&& \ar[d]^{\mathbb{I}_A\times f_{N(k_1,k_2),k_3}\times\mathbb{I}_B}
A\times \PP^{N(k_1,k_2)}\times \PP^{k_3}\times B\\
A\times\PP^{k_1}\times \PP^{N(k_2,k_3)}\times B\ar[rrr]^{\mathbb{I}_A\times  f_{k_1,N(k_2,k_3)}\times \mathbb{I}_B}&&&
A\times \PP^{N(k_1,k_2,k_3)}\times B,
}
\]
where $A$ and $B$ are products of projective spaces. Therefore we can apply Lemma \ref{pullback}
to conclude that $z$ lives in $w'$.

\textit{Case} $j<i+1$. In this case, the above commutative square represents morphisms of the form 
\[
\xymatrix{\ar[d]
A\times \PP^{k_1}\times \PP^{k_2}\times B\times  \PP^{k_3}\times \PP^{k_4}\times C\ar[r] & \ar[d]
A\times \PP^{N(k_1,k_2)}\times B\times \PP^{k_3}\times \PP^{k_4}\times C\\
A\times\PP^{k_1}\times \PP^{k_2}\times B\times \PP^{N(k_3,k_4)}\times C\ar[r]&
A\times \PP^{N(k_1,k_2)}\times B\times \PP^{N(k_3,k_4)}\times C
}.
\]
Since $z$ lives in $w$ we may decompose $z=(z_A, z_{12}, z_B,z_{34},z_C)$.
Since it lives in $v$, we have $z_{12}\in \Sigma_{N(k_1,k_2)}$
and since it lives in $v'$ we have $z_{34}\in \Sigma_{N(k_3,k_4)}$.
It directly follows that $z$ lives in $w'$.
\end{proof}

\begin{theo}[Generalized Decomposability]
Let $n\geq 2$ and $1<q\leq n$ be integers.
A point $z\in \PP^{N_n}$ is $q$-decomposable if and only if 
it is in at least $q-1$ different Segre varieties of the form 
$\Sigma_{N_\ell,N_{n-\ell}}$, with $0\leq \ell\leq n$.
\end{theo}
\begin{proof} Using the cubical representation introduced above, it suffices to show that if
 $z$ lives in $q-1$ different vertices $v^1,\cdots,v^{q-1}$ of degree
$|v^i|=n-2$, then $z$ lives in a vertex of degree $(n-q)$.
By recursively applying Lemma \ref{pullback_living} we find that $z$ lives in the vertex $v^*$ of degree $(n-q)$
whose components are
$v^*_j=\mathrm{min}_i \{v_j^i\}.$
This implies that $z$ is $q$-decomposable.
The converse statement is trivial.
\end{proof}

\bibliographystyle{apsrev4-1}
\bibliography{bibliografia}

\end{document}